\documentclass[12pt,reqno]{article}

\usepackage{cite,epic,eepic,euscript,verbatim,amsmath,amsthm,amssymb,amsfonts,afterpage,float,bm,stmaryrd,paralist,cancel,color,fancyhdr,graphics,graphicx,epsf,hyperref,makeidx,epsfig}

\usepackage{subfigure}
\usepackage{array}
\newtheorem{theorem}{Theorem}[section]
\newtheorem{lemma}{Lemma}[section]

\numberwithin{equation}{section}
\numberwithin{table}{section}
\numberwithin{figure}{section}

\newcommand{\reff}[1]{{\rm (\ref{#1})}}

\newcommand{\bbR}{\mathbb{R}}            
\newcommand{\R}{\mathbb{R}}            

\def\XXint#1#2#3{{\setbox0=\hbox{$#1{#2#3}{\int}$}
\vcenter{\hbox{$#2#3$}}\kern-.51\wd0}}

\begin{document}


\title{A Modified Poisson--Nernst--Planck Model with Excluded Volume Effect: Theory and Numerical Implementation}

\author{Farjana Siddiqua\thanks{
Department of Mathematics and Statistics, Florida International University, Miami, FL, U. S. A.}
\and
Zhongming Wang\thanks{
Department of Mathematics and Statistics, Florida International University, Miami, FL, U. S. A.}
\and
Shenggao Zhou\thanks{
Department of Mathematics and Mathematical Center for Interdiscipline Research, Soochow University, 1 Shizi Street, Suzhou 215006, Jiangsu, China.
Corresponding author. E-mail: sgzhou@suda.edu.cn.
}
}
  \pagestyle{myheadings} \markboth{Modified PNP Model with Excluded Volume Effect}{F. Siddiqua, Z. Wang, and S. Zhou}
  
  \date{}
  
\maketitle

\begin{abstract}
The Poisson--Nernst--Planck (PNP) equations have been widely applied to describe ionic transport in ion channels, nanofluidic devices, and many electrochemical systems. Despite their wide applications, the PNP equations fail in predicting dynamics and equilibrium states of ionic concentrations in confined environments, due to the ignorance of the excluded volume effect. In this work, a simple but effective modified PNP (MPNP) model with the excluded volume effect is derived, based on a modification of diffusion coefficients of ions. At the steady state, a modified Poisson--Boltzmann (MPB) equation is obtained with the help of the Lambert-W special function.  The existence and uniqueness of a weak solution to the MPB equation are established. Further analysis on the limit of weak and strong electrostatic potential leads to two modified Debye screening lengths, respectively.  A numerical scheme that conserves total ionic concentration and satisfies energy dissipation is developed for the MPNP model.  Numerical analysis is performed to prove that our scheme respects ionic mass conservation and satisfies a corresponding discrete free energy dissipation law. Positivity of numerical solutions is also discussed and numerically investigated. Numerical tests are conducted to demonstrate that the scheme is of second-order accurate in spatial discretization and has expected properties. Extensive numerical simulations reveal that the excluded volume effect has pronounced impacts on the dynamics of ionic concentration and flux. In addition, the effect of volume exclusion on the timescales of charge diffusion is systematically investigated by studying the evolution of free energies and diffuse charges.
 \bigskip

\end{abstract}

{\allowdisplaybreaks
\section{Introduction}
Ionic transport has been observed in a wide variety of technological applications and biological processes, such as membrane ion channels, electrochemical energy devices, and electrokinetics in microfluidics\cite{Hille_Book2001, LiBook_ImperialCollege, Schoch_RMP08, DLi_Book04}. Based on a mean-field approximation, the Poisson--Nernst--Planck (PNP) equations can be derived to describe the dynamics of ions under an electric field. The Nernst-Planck equations model the diffusion and migration of ions in gradients of ionic concentrations and electrostatic potential.  The Poisson's equation governs electrostatic potential with the charge density stemming from transporting ions.

Despite its success in many applications, the PNP theory fails in predicting dynamics and equilibrium states of ionic concentrations. One reason behind this is its ignorance of the excluded volume effect, which is of importance in confined environments.  The crucial effect of excluded volume is able to prevent unphysical crowding of pointlike counterions at the vicinity of charged surface by forming a compactly packed layer of hydrated counterions, called Stern layer.  As such, there is a saturation concentration of counterions near a charged surface due to steric hindrance. With less adsorbed counterions, the Debye screening length therefore increases\cite{ZhouWangLi_PRE11, LiLiuXuZhou_Nonliearity13}.  In addition, the excluded volume effect has profound impacts on the dynamics of ionic transport. Analysis on current-voltage relations of an ion channel shows that the excluded volume effect reduces ionic current inside a channel\cite{BZLu_JCP14}.  Nonlinear modification of ionic mobility due to the excluded volume effect leads to saturation of current through an ionic channel on account of overcrowding of ions\cite{Burger_Nonlinearity12}.

At the steady state, the PNP equations are reduced to the classical Poisson--Boltzmann (PB) equation, if zero-flux boundary conditions are imposed. Recently, there has been  growing interests in incorporating the excluded volume effect to such mean-field models. The classical PB theory has been modified to study the excluded volume effect on the equilibrium distribution of ions in charged systems\cite{BAO_PRL97, BAO_electroacta00,ChuEtal_BiophysJ07,Li_SIMA09,Li_Nonlinearity09}.   Within the framework of the PNP theory, several versions of modified PNP theory with the excluded volume effect have been proposed to describe ionic transport. One common approach is to add an excess chemical potential to the potential of mean force. Such a correction is able to address the excluded volume effect\cite{BazantSteric_PRE07, BZLu_BiophyJ11, BZLu_JCP14, BZLu_JSP16, Burger_Nonlinearity12, GLin_Cicp14, JZWu_JPCM14, LiuTu_JDDE12, LiuTuZhang_JDDE12, LinLiuZhang_SIADS13, JiaLIuZhang_DCDSB16,  BobHyonLiu_JCP10, HyonLiuBob_CMS10, HyonLiuBob_JPCB12, LinBob_CMS14, XuShengLiu_CMS14}, dielectric effects\cite{MaXu_JCP14,LiWenZhou_CMS16, DChen_BMB16, LiuMaXu_Cicp17}, and ion-ion correlations\cite{XuMaLiu_PRE14, LiuJiXu_SIAP17}.  For instance, the excluded volume effect is included by considering the entropy of solvent molecules, giving rise to a model with nonlinearly modified mobilities\cite{BazantSteric_PRE07, BZLu_BiophyJ11, BZLu_JCP14}. A more sophisticated strategy is to incorporate the excluded volume effect by adding an excess chemical potential, which is described by the density functional theory (DFT)\cite{LiuTu_JDDE12, LiuTuZhang_JDDE12, GLin_Cicp14, JZWu_JPCM14}, or by the Lennard-Jones potential accounting for hard-sphere repulsions\cite{HyonLiuBob_CMS10, BobHyonLiu_JCP10,LinBob_CMS14}. To avoid computationally intractable integro-differential equations, local approximations of nonlocal integrals are employed to obtain local models\cite{HyonLiuBob_CMS10, HyonLiuBob_JPCB12, LinBob_CMS14, LinLiuZhang_SIADS13,JiaLIuZhang_DCDSB16,BZLu_JSP16}.

In this work, we develop a simple but effective modified PNP (MPNP) model, following the treatment of the excluded volume effect proposed for the diffusion of hard spheres\cite{Bruna_PRE12}. The excluded volume effect introduces a modification of diffusion coefficients depending linearly on ionic concentrations. At the steady state, a modified Poisson--Boltzmann (MPB)  equation is derived by using the principal branch of the Lambert function. Further analysis establishes the existence and uniqueness of a weak solution to the MPB equation. In the limit of weak electrostatic potential, the MPB equation is linearized as the Debye-H\"uckel equation with a modified Debye screening length, which is longer than the classical one due to the excluded volume effect. Such a result agrees with previous models with volume exclusion, and is later confirmed by numerical simulations.  When strong electrostatic potential is considered, the MPB equation is reduced to a linear equation with a different screening length modified by volume exclusion.   Our numerical results illustrate that the MPNP model is capable of capturing the effect of volume exclusion on equilibrium ionic distributions and the timescale of charge diffusion.

Due to nonlinear coupling of electrostatic potential and ionic concentrations, it is not trivial to solve the PNP equations analytically, even numerically. Many numerical methods have been proposed in the literature. A hybrid numerical scheme that uses adaptive grids was developed to solve the PNP equations in two dimensions \cite{Gibou_JCP14}. A second-order accurate finite difference scheme was proposed to discretize the PNP equations with three important properties, which are total ionic conservation, energy dissipation, and solution positivity\cite{XFLi_JCE14}. Recently, a delicate temporal discretization scheme was designed to preserve energy dynamics\cite{XFLi_JCE17}.  By using Slotboom variables, Liu and Wang \cite{LW14} developed a free energy satisfying finite difference scheme that respects those three properties. They also constructed a free energy satisfying discontinuous Galerkin method, in which the positivity of numerical solutions is enforced by an accuracy-preserving limiter\cite{LiuWangDG_JCP17}.  He and Pan\cite{HePan_AMC16} designed a finite-difference discretization for the 2D PNP equations, which conserves total concentration and preserves electrostatic potential energy.  A finite element discretization that can enforce positivity of numerical solutions was developed for the PNP equations, as well as the PNP equations coupling with the incompressible Navier-Stokes equations\cite{XuLiu_JCP16}.

In contrast to the classical PNP equations, not much progress has been made on the development of numerical methods that can guarantee physical properties respected by the MPNP equations with the excluded volume effect. Chaudhry \emph{et al.} \cite{Chaudhry_Cicp14} developed a stabilized, mass-conserving finite element for modified PNP equations with the excluded volume effect. Tu \emph{et al.}\cite{BenLu_CPC15}  proposed a finite element method with stabilized techniques to enhance numerical robustness in solving modified PNP equations with the excluded volume effect in 3D.  In this work, we propose a finite difference scheme for the newly derived MPNP equations with the excluded volume effect. We prove that the numerical solution has desired properties that total ionic concentrations conserve and the discrete free energy dissipates monotonically. In addition, we prove the positivity of numerical solutions for a one-dimensional case that has many realistic applications. We also discuss several issues involving proving positivity of numerical solutions in high dimensions.

The rest of the paper is organized as follows. In section \ref{s:Model}, we derive a modified PNP model with with the excluded volume effect and analyze the steady state of the model.  In section \ref{s:NumMethod}, we detail the algorithm of our numerical method for the derived MPNP model, and prove properties of our numerical method. Section \ref{s:NumSim} is devoted to showing our numerical results. Finally, we draw our conclusions in section \ref{s:Con}.

\section{Model}\label{s:Model}
We consider an ionic solution occupying a bounded domain $\Omega$ in $\bbR^d$ with a boundary $\partial \Omega$, and $d=2,3$.  We assume there are $M$ ionic species in the solution. Denote by $\psi: \Omega \to \bbR$ the electrostatic potential, and $c^l=c^l(t,\cdot)$ the local ionic concentration of the $l^{th}$ species at time $t$ and a spatial point. Note that the electrostatic potential is also a time-dependent function, since it couples with time-dependent ionic concentrations through the Poisson's equation.  Denote by $q_l$ the valence of the $l^{th}$ ionic species.

We consider the diffusion of ions with the excluded volume effect under the gradient of given electrostatic potential. It is known that diffusion coefficients of ions have sensitive dependence on sizes of ions\cite{EVE_JCIS08}.  Electrophoretic mobilities predicted by classical models that ignore the excluded volume effect are smaller than that measured by experiments\cite{Lyklema_JCIS07}.  The collective mobilities of ions are enhanced by the inclusion of excluded volume effect. We know by the Einstein relation that diffusion coefficients of ions increase as well. In this work, we assume that diffusion coefficients of ions are functions of local ionic concentrations.  We have
the following modified Nernst--Planck equations for ionic concentrations in a dimensionless form:
\begin{equation}\label{MNPEqns}
 \partial_t c^l  =\gamma\nabla \cdot \left( D^l(c^l) \nabla  c^l+  q_l  c^l \nabla \psi  \right), \quad l=1,2,\cdots, M,
\end{equation}
where $\gamma$ is a positive coefficient from nondimensionalization and $D^l(c^l)$ is the diffusion coefficient for $c^l$. The electrostatic potential $\psi$ is governed by the Poisson's equation
\begin{equation}
- \epsilon\Delta \psi  = \sum_{l=1}^Mq_lc^l,
\end{equation}
where $\epsilon$ is a positive dimensionless parameter. For different applications, different boundary conditions can be imposed for electrostatic potential. For instance, Dirichlet boundary conditions are prescribed to represent fixed electrostatic potential on the boundary, Neumann boundary conditions are imposed to describe surface charge density on the boundary, and Robin boundary conditions can be used to describe surface capacitance.

 In the literature, many modified Poisson--Nernst--Planck models with the excluded volume effect can be regarded as modification of diffusion coefficients with respect to ionic concentrations\cite{BazantSteric_PRE07, BZLu_BiophyJ11, Burger_Nonlinearity12, BZLu_JSP16}.  By the method of matched asymptotic expansions, Bruna and Chapman \cite{Bruna_PRE12} derive a linear functional dependence of diffusion coefficients on ionic concentrations, to account for the excluded volume effect.  Following this treatment of the excluded volume effect, we use
$$D^l(c^l)=1+\alpha_l c^l, $$
where $\alpha_l$ is a size-related positive parameter arising from volume exclusion interactions.  In summary, we have the modified Poisson--Nernst--Planck equations
\begin{equation}\label{SMPNPi}
\left\{
\begin{aligned}
& \partial_t c^l  =\gamma\nabla \cdot \left( \nabla  c^l+\alpha_l c^l \nabla c^l+  q_l  c^l \nabla \psi  \right), \quad l=1,2,\cdots, M, \quad  x\in  \Omega, \; t>0, \\
&- \epsilon\Delta \psi  = \sum_{l=1}^Mq_lc^l,  \quad x\in  \Omega, \; t>0.
\end{aligned}
\right.
\end{equation}


\subsection{Related Models}
We discuss several related models with the excluded volume effect. By incorporating entropies of solvent molecules, a type of MPNP models with concentration-dependent diffusion coefficients has also been developed\cite{BazantSteric_PRE07, BZLu_BiophyJ11, Burger_Nonlinearity12}. The diffusion coefficient for each ionic species is a nonlinear function of concentrations of all ionic species. Another related model has been developed by using local approximations of the Lennard-Jones potential for hard-sphere interactions\cite{HyonLiuBob_CMS10, HyonLiuBob_JPCB12, LinBob_CMS14}. The corresponding modified Nernst-Planck equation is given by
\[
 \partial_t c^l  =\gamma\nabla \cdot \left( \nabla c^l+  q_l  c^l \nabla \psi + \sum_{k=1}^M g_{lk} c^l \nabla c^k \right),
\]
where $g_{lk}$ are positive constants related to ionic sizes.  It is shown that the corresponding free energy is strictly convex if and only if the matrix ${\bf G}:= (g_{lk})$ is positive semi-definite\cite{Gavish_arXiv17}. Our model corresponds to zero off-diagonal entries of ${\bf G}$, in which case the MPNP system is asymptotically stable and does not have multiple steady states\cite{LinBob_CMS14, LinBob_Nonlinearity15, Gavish_arXiv17}.  The off-diagonal entries of ${\bf G}$ should be carefully chosen when cross diffusion of different ionic species is taken into account. 

\section{Modified Poisson--Boltzmann Equation}\label{s:MPB}
We investigate the excluded volume effect on the steady state of the MPNP equations \reff{SMPNPi}. To focus on studying our treatment of volume exclusion, we simply use Dirichlet boundary conditions for the electrostatic potential, i.e., $\psi =\psi_B $ on $\partial \Omega$. From \reff{SMPNPi}, we obtain equilibrium distributions of concentrations in terms of the electrostatic potential:
\[
c^l(\psi)=\alpha_l^{-1} W_0\left(\alpha_l \eta_l e^{-q_l \psi}  \right),
\]
where $W_0(\cdot)$ is the principal branch of the Lambert function\cite{Lambert_ACM96}, and $\eta_l$ is a positive constant determined by $\eta_l=c_{\infty}^l e^{\alpha_l c_{\infty}^l}$. Here, $c_{\infty}^l$ is the ionic concentration when the electrostatic potential vanishes.  As such, we have a modified Poisson--Boltzmann (MPB) equation
\begin{equation}\label{MPB}
- \epsilon\Delta \psi  = \sum_{l=1}^Mq_l \alpha_l^{-1} W_0\left(\alpha_l \eta_l e^{-q_l \psi}  \right) \quad \mbox{with}~ \psi =\psi_B  ~\mbox{on}~ \partial \Omega.
\end{equation}
Following the notation used in Refs. \cite{Li_Nonlinearity09,LiLiuXuZhou_Nonliearity13}, we define
\[
V(\phi)= - \sum_{l=1}^M q_l \int_0^\phi c^l(\zeta) d\zeta, ~  \zeta \in \R.
\]
\begin{lemma}\label{l:V}
The function $V: \R \to \R$ is a $C^\infty$ function. Moreover, it is a strictly convex function that has a bounded second derivative, $\displaystyle \mbox{Min}_{\phi \in \R} V(\phi) = V(0) =0$, $V'(0)=0$, and $\lim_{\phi \to \pm \infty} V(\phi) = + \infty$.
\end{lemma}

\noindent
{\bf Proof.}
Since $W_0(u)$ is an analytical function for $u>0$, it is easy to show $V(\cdot)$ is a $C^\infty$ function. Now we verify that
\[
V'(0)= -  \sum_{l=1}^M q_l c^l (0) = -  \sum_{l=1}^M q_l \alpha_l^{-1} W_0\left(\alpha_l \eta_l \right) = -  \sum_{l=1}^M q_l  c_{\infty}^l = 0,
\]
where we use the bulk neutrality condition in the last equation. Also, we have
\begin{align*}
V''(\phi) = -  \sum_{l=1}^M q_l \left[c^l (\phi)\right]' =  \sum_{l=1}^M q_l^2 \eta_l  e^{-q_l \phi} W_0^{'}\left(\alpha_l \eta_l e^{-q_l \phi}\right)
             =  \sum_{l=1}^M q_l^2 \alpha_l^{-1} \frac{W_0\left(\alpha_l \eta_l e^{-q_l \psi}  \right)}{1+W_0\left(\alpha_l \eta_l e^{-q_l \psi}  \right)},
\end{align*}
where in the last equation we use the identity
$$W'(u)=\frac{W_0(u)}{u\left[1+W_0(u)\right]}.$$
It is easy to see that
$$0<V''(\phi)< \sum_{l=1}^M q_l^2 \alpha_l^{-1}.$$ Therefore, $V(\phi)$ achieves its minimum value $V(0)=0$, and $V'(\phi) >0$ for $\phi>0$ and $V'(\phi) <0$ for $\phi<0$. Simple calculations can verify that $\lim_{\phi \to \pm \infty} V(\phi) = + \infty$. \qed

We now consider the existence of a weak solution to the boundary value problem \reff{MPB}. We use standard notation for Sobolev spaces\cite{GilbargTrudinger98}. Let $$H_{\rm \psi_B}^1(\Omega) = \left\{ \phi \in H^1(\Omega):\phi=\psi_{B}~\text{on}~\partial \Omega \right\}. $$
\begin{theorem}\label{t:Solun}
Let $\Omega$ be a nonempty, bounded, and open subset of $\R^3$. Assume the boundary $\partial \Omega$ is of $C^2$. There exists a unique weak solution $\psi \in H_{\rm \psi_B}^1(\Omega) \cap L^\infty (\Omega)$ to the boundary value problem \reff{MPB}.
\end{theorem}
\noindent
{\bf Proof.} From Lemma \ref{l:V}, we know that the Theorem 2.1 given in \cite{Li_SIMA09} (and a correction of the proof in \cite{Li_SIMA11}) apply to our case. We therefore omit the proof here.  \qed

To explore more about the MPB model, we consider the limit of weak electrostatic potential, which gives a modified Debye screening length due to the  excluded volume effect.
\begin{theorem}\label{t:DebyeLen}
In the limit of weak electrostatic potential, i.e.,  $|\psi| \ll 1$,  the modified Debye screening length is given by
$$\hat{\lambda}_D^W =\left[\sum_{l=1}^M \frac{q_l^2  c_{\infty}^l }{\epsilon \left(1+\alpha_l c_{\infty}^l\right) }  \right]^{-\frac{1}{2}}.$$
\end{theorem}
\noindent
{\bf Proof.}
By Taylor expansions, we have for  $|\psi| \ll 1$ that
\begin{align*}
-\epsilon \Delta \psi &=  \sum_{l=1}^Mq_l \alpha_l^{-1} W_0\left(\alpha_l \eta_l e^{-q_l \psi}  \right) \\
                                &=   \sum_{l=1}^Mq_l \alpha_l^{-1} \left[ W_0\left(\alpha_l \eta_l \right) -\alpha_l \eta_l q_l W_0^{'}\left(\alpha_l \eta_l \right) \psi +   O(\psi^2)   \right] \\
                                &=  \sum_{l=1}^Mq_l \alpha_l^{-1}  W_0\left(\alpha_l \eta_l \right) - \sum_{l=1}^M \frac{q_l^2 \alpha_l^{-1} W_0\left(\alpha_l \eta_l \right) }{1+W_0\left(\alpha_l \eta_l \right)} \psi  + O(\psi^2) \\
                                &=  \sum_{l=1}^M q_l  c_{\infty}^l  - \sum_{l=1}^M \frac{q_l^2  c_{\infty}^l }{1+\alpha_l c_{\infty}^l } \psi  + O(\psi^2).
\end{align*}
Ignoring $O(\psi^2)$ terms, we have the Debye-H\"uckel equation
\[
\Delta \psi =(\hat{\lambda}_D^{W})^{-2} \psi
\]
with the Debye screening length $\hat{\lambda}_D^W =\left[\sum_{l=1}^M \frac{q_l^2  c_{\infty}^l }{\epsilon \left(1+\alpha_l c_{\infty}^l\right) }  \right]^{-\frac{1}{2}}$. This completes the proof. \qed

We remark that, in contrast to the classical Debye screening length $\lambda_D =\left(\sum_{l=1}^M q_l^2  c_{\infty}^l / \epsilon \right)^{-\frac{1}{2}}$, the excluded volume effect leads to a longer modified Debye screening length.  This result agrees with other PB models with volume exclusion\cite{LiLiuXuZhou_Nonliearity13}. In our numerical simulations, we observe that less counterions are adsorbed to charged surface on account of the excluded volume effect, giving rise to higher surface electrostatic potential. This indicates that the Debye screening length becomes longer.

Near charged surface, it is of practical interest to study the behavior of counterions.  Denote by $\mathcal J$ the set of indice for counterions speices.  It is reasonable to assume that the electrostatic potential near surface has an opposite sign to the counterions, i.e., $q_l \psi <0$ for $l\in \mathcal J$.  When the surface potential is strong ($|\psi| \gg 1$), we consider the limit that $e^{-q_l \psi} \gg 1$. From an asymptotic approximation that $W_0(u) \sim \ln(u)$ for large positive $u$, we have by keeping leading order terms that
\begin{align*}
-\epsilon \Delta \psi &=  \sum_{l\in \mathcal J} q_l \alpha_l^{-1} \left[\ln\left(\alpha_l \eta_l\right) -q_l \psi \right].
\end{align*}
We rewrite it in the form
\begin{align*}
\Delta \psi =(\hat{\lambda}_D^{S})^{-2} \psi+ R,
\end{align*}
where $\hat{\lambda}_D^{S} =\left(\sum_{l\in \mathcal J}  \frac{q_l^2}{\epsilon \alpha_l}  \right)^{-\frac{1}{2}}$ and the constant $R=-\sum_{l\in \mathcal J}  \frac{q_l^2}{\epsilon \alpha_l} \ln \left (\alpha_l \eta_l \right)$.
It is interesting to see that, in the strong limit of electrostatic potential, the leading order terms of the MPB \reff{MPB} becomes an equation resembling the Debye--H\"uckel equation with a constant charge source arising from the bulk. The corresponding screening length $\hat{\lambda}_D^{S}$ depends on the parameters $\alpha_l$ arising from the excluded volume effect, rather than bulk concentrations.
\section{Dynamics}\label{s:dynamics}
In this and following sections, we study the dynamics of ionic concentrations and electrostatic potential in a closed system that has an impenetrable boundary with certain surface charge density. We focus on the physical properties of the system, and develop a suitable numerical scheme to capture the properties discretely. The corresponding discrete properties are established and confirmed by numerical simulations. 

To model the closed system with boundary surface charge, we use zero-flux boundary conditions for ionic concentrations:
\[
\left( \nabla  c^l+\alpha_l c^l \nabla c^l+  q_l  c^l \nabla \psi  \right)\cdot\textbf{n} =0 \quad  \mbox{on}~ \partial \Omega,
\]
and Neumann boundary conditions for the electrostatic potential:
\[
\epsilon \nabla\psi \cdot\textbf{n} =\sigma \quad  \mbox{on}~ \partial \Omega.
\]
Here $\textbf{n}$ is the exterior unit normal vector, and $\sigma$ is the surface charge density. The initial conditions,
\[
c^l(x, 0) = c^l_{\rm in}(x),
\]
are set to satisfy the neutrality condition
\[
\int_{\partial \Omega}\sigma dS + \sum_{l=1}^M\int_{\Omega}  q_l c_{\rm in}^l dx  =0,
\]
which is necessary for solvability of the problem. In summary, we study the following initial-boundary value problem
\begin{equation}\label{IBVP}
\left\{
\begin{aligned}
 & \partial_t c^l  =\gamma\nabla \cdot \left( \nabla  c^l+\alpha_l c^l \nabla c^l+  q_l  c^l \nabla \psi  \right), \quad l=1,2,\cdots, M, \quad  x\in  \Omega, \; t>0, \\
 & - \epsilon\Delta \psi  = \sum_{l=1}^Mq_lc^l,  \quad x\in  \Omega, \; t>0, \\
 & c^l(0, x)  =c^l_{\rm in}(x),  \quad x\in \Omega,\\
 & \epsilon \nabla\psi \cdot\textbf{n} =\sigma, \quad  \left( \nabla  c^l+\alpha_l c^l \nabla c^l+  q_l  c^l \nabla \psi  \right)\cdot\textbf{n} =0,\quad
 \quad x\in \partial\Omega,\; t>0.
\end{aligned}
\right.
\end{equation}


Since $c^l$ represents concentrations of ions, it is reasonable to assume that $c^l(t,x) > 0$ for $x\in \Omega$ and $t>0$. By zero-flux boundary conditions and the Nernst-Planck equations, we have ionic mass conservation in the sense that
\begin{equation*}
\begin{aligned}
\frac{d}{dt}\int_{\Omega}c^l(t, x)dx &=\int_{\Omega} \gamma\nabla \cdot \left( \nabla  c^l+\alpha_l c^l \nabla c^l+  q_l  c^l \nabla \psi  \right) dx  \\
                                                       &=\int_{\partial \Omega} \gamma \left( \nabla  c^l+\alpha_l c^l \nabla c^l+  q_l  c^l \nabla \psi  \right)\cdot\textbf{n} dS =0.
\end{aligned}
\end{equation*}
For the MPNP model \reff{SMPNPi}, we propose the following total free energy
\[ F =  \sum_{l=1}^M\int_{\Omega}   c^l\ln c^l dx + \frac{1}{2} \sum_{l=1}^M\int_{\Omega} \alpha_l {(c^l)}^2 dx +\frac{1}{2}\sum_{l=1}^M \int_{\Omega}   q_lc^l\psi dx + \frac{1}{2}\int_{\partial\Omega} \sigma \psi dS,
\]
where the first term represents entropic contributions, the second term is the ionic interaction energy due to volume exclusion, and the third and fourth terms are the electrostatic energies.  We consider time evolution of the free energy
\begin{equation*}
\begin{aligned}
\frac{d F}{dt} &=  \sum_{l=1}^M\int_{\Omega} \partial_t c^l \left( \ln c^l +1 + \alpha_l c^l +q_l \psi \right) dx \\
                       &\quad+\sum_{l=1}^M \int_{\Omega}   \frac{1}{2} q_l \partial_t c^l\psi + \frac{1}{2} q_l  c^l \partial_t\psi  -q_l \partial_t c^l  \psi  dx +\frac{1}{2}\int_{\partial\Omega} \partial_t \sigma \psi  + \sigma \partial_t \psi dS \\
                      &=  \sum_{l=1}^M\int_{\Omega} \gamma \left[ \nabla \cdot \left( \nabla  c^l+\alpha_l c^l \nabla c^l+  q_l  c^l \nabla \psi  \right) \right] \left( \ln c^l +1 + \alpha_l c^l +q_l \psi \right) dx \\
                       &\quad+ \int_{\Omega}   -\frac{1}{2} \epsilon \partial_t \psi \Delta \psi  + \frac{1}{2} \epsilon \Delta (\partial_t \psi) \psi   dx +\frac{1}{2}\int_{\partial\Omega} \partial_t \sigma \psi  + \sigma \partial_t \psi dS \\
                       &= - \sum_{l=1}^M \int_{\Omega}  \frac{\gamma}{ c^l} \left| \nabla c^l+ \alpha_l c^l \nabla c^l + q_lc^l\nabla \psi \right |^2 dx + \int_{\partial\Omega}  \partial_t\sigma  \psi dS.
\end{aligned}
\end{equation*}
Assuming that the surface charge density is time independent, we have free energy dissipation law $\frac{d F}{dt} \leq 0$.
In summary, we assume the following three dynamical properties for any solution to (\ref{SMPNPi}):
\begin{subequations}
\begin{align}
&\mbox{(P1):}  \qquad c^l (x,t) > 0 \quad \mbox{for}~ x\in \Omega ~\mbox{and}~  t>0,\label{str1} \\
&\mbox{(P2):}  \qquad \int_{\Omega} c^l(t, x )\,dx=\int_{\Omega} c^l_{\rm in}(x)\,dx \quad \mbox{for}~ t>0,\label{str2}\\
&\mbox{(P3):}  \qquad \frac{d}{dt}   F  \leq 0 \quad \mbox{for}~ t>0. \label{str3}
\end{align}
\end{subequations}
\section{Numerical Method}\label{s:NumMethod}
\subsection{Reformulation}
For conciseness we present our method in $\R^2$, while the algorithm can be extended to $\R^3$ in a dimension by dimension manner.
We formally reformulate the system by using Slotboom variables \cite{LW14}
$$g^l(t,x,y)=c^l(t,x,y)e^{q_l\psi(t,x,y)+\alpha_l c^l(t,x,y)},$$
to obtain
the following two sets of equations
{ 
\begin{align}
&c^l_t  =\gamma(e^{-(q_l\psi+\alpha_l c^l)}g^l_x)_x +\gamma(e^{-(q_l\psi+\alpha_l c^l)}g^l_y)_y,\\
&-\epsilon(\psi_{xx}+\psi_{yy}) =\sum_{l=1}^Mq_lc^l.\label{Poisson2d}
\end{align}

We now describe our algorithm by first  partitioning the square domain $[a,b]\times[a',b']$ with a uniform partition of $x_i=a+h(i-1/2)$ and $y_j=a'+h(j-1/2)$ for  $i=1,\cdots,N_x$ and $j=1,\cdots,N_y$.
\subsection{Algorithm} \label{algorithm}
\begin{itemize}
\item[1.] We use $c^l_{i,j}$ to approximate $c^l(t, x_i,y_j)$ and $\psi_{i,j}$ to approximate $\psi(t, x_i,y_j)$.  Given $c^l_{i,j}, i=1, \cdots,  N_x$, $j=1, \cdots,  N_y$, we compute the potential $\psi_{i,j}$ by
\begin{equation}\label{pc}
-\epsilon\frac{\psi_{i+1,j}-2\psi_{i,j}+\psi_{i-1,j}}{h^2}-\epsilon\frac{\psi_{i,j+1}-2\psi_{i,j}+\psi_{i,j-1}}{h^2}=\sum_{l=1}^Mq_lc^l_{i,j},
\end{equation}
where $\psi_{i,1}-\psi_{i,0}=-\sigma_{i,1/2} h/\epsilon$, $\psi_{i,N+1}-\psi_{i,N}=\sigma_{i,N+1/2} h/\epsilon$, $\psi_{1,j}-\psi_{0,j}=-\sigma_{1/2,j} h/\epsilon$, and $\psi_{N+1,j}-\psi_{N,j}=\sigma_{N+1/2,j} h/\epsilon$. Here $\sigma_{i,1/2}$, $\sigma_{i,N+1/2}$, $\sigma_{1/2,j}$ and $\sigma_{N+1/2,j}$ are boundary conditions at $y=a'$, $y=b'$,$x=a$ and $x=b$, respectively.   
For definiteness,  we  set  $\psi_{1,1}=0$ at any time $t$ to single out  a particular solution since $\psi$ is unique up to an additive constant.
\item[2.] With the above obtained $\psi_{i,j}, i=1,\cdots,N_x$, $j=1,\cdots, N_y$, the semi-discrete approximation of  the concentration $c^l$ satisfies
\begin{align}\label{tc}
\frac{d}{dt}c^l_{i,j}&=\frac{\gamma}{h} \left[e^{-(q_l\psi_{i+{\frac{1}{2}},j}+\alpha_l c^l_{i+{\frac{1}{2}},j})}\widehat{g^l}_{x,i+{\frac{1}{2}},j} -e^{-(q_l\psi_{i-{\frac{1}{2}},j}+\alpha_l  c^l_{i-{\frac{1}{2}},j})}\widehat{g^l}_{x,i-{\frac{1}{2}},j}\right] \nonumber\\
&+\frac{\gamma}{h} \left[
e^{-(q_l\psi_{i,j+{\frac{1}{2}}}+\alpha_l  c^l_{i,j+{\frac{1}{2}}})}\widehat{g^l}_{y,i,j+{\frac{1}{2}}} -e^{-(q_l\psi_{i,j-{\frac{1}{2}}}+\alpha_l  c^l_{i,j-{\frac{1}{2}}})}\widehat{g^l}_{y,i,j-{\frac{1}{2}}}\right]:=Q_{i,j}(c^l,\psi),
\end{align}
where
\begin{align*}
\psi_{i+{\frac{1}{2}},j}&=\frac{\psi_{i+1,j}+\psi_{i,j}}{2},
\quad \psi_{i,j+{\frac{1}{2}}}=\frac{\psi_{i,j+1}+\psi_{i,j}}{2}, \notag\\
c^l_{i+{\frac{1}{2}},j}&=\frac{c^l_{i+1,j}+c^l_{i,j}}{2},
\quad c^l_{i,j+{\frac{1}{2}}}=\frac{c^l_{i,j+1}+c^l_{i,j}}{2}, \notag\\
\widehat{g^l}_{x,i+{\frac{1}{2}},j}&=\frac{g^l_{i+1,j}-g^l_{i,j}}{h}=\frac{c^l_{{i+1,j}}e^{q_l\psi_{i+1,j}+\alpha_l c^l_{i+1,j}}-c^l_{i,j}e^{q_l\psi_{i,j}+\alpha_l c^l_{i,j}}}{h},\notag\\
\widehat{g^l}_{y,i,j+{\frac{1}{2}}}&=\frac{g^l_{i,j+1}-g^l_{i,j}}{h}=\frac{c^l_{{i,j+1}}e^{q_l\psi_{i,j+1}+\alpha_l c^l_{i,j+1}}-c^l_{i,j}e^{q_l\psi_{i,j}+\alpha_l c^l_{i,j}}}{h},\notag\\
\widehat{g^l}_{x, 1/2,j}&=0, ~ \widehat{g^l}_{x, N_x+1/2,j}=0, ~ \widehat{g^l}_{y, i,1/2}=0, ~\mbox{and }~ \widehat{g^l}_{y,i,N_y+1/2}=0.
\end{align*}
\item[3.] Discretize $t$ uniformly and let $t_n=t_0+kn$, $c^{l,n}_{i,j} \sim c(t_n,x_i,y_j)$ and $\psi_{i,j}^n \sim \psi(t_n,x_i,y_j)$, we then solve \eqref{Poisson2d} by
\begin{equation}\label{cq}
\frac{c_{i,j}^{l,n+1}-c_{i,j}^{l,n}}{k}=Q_{i,j}(c^{l,n},\psi^n).
\end{equation}
\end{itemize}

\subsection{Numerical Properties}\label{s:Properties}
In this section we investigate the properties of our algorithm.  We will show the desired properties, such as conservation and free energy dissipation for our Algorithm \ref{algorithm} in the following.
\begin{theorem}\label{thm:properties}  Let $c^l_{i,j}=c^l(t, x_i, y_j)$ and  $\psi_{i,j}=\psi(t, x_i, y_j)$ be semi-discrete solutions from \eqref{tc} and \eqref{pc} respectively;  and $c^{l,n}_{i,j}=c^l(t_n, x_i, y_j)$ and  $\psi_{i,j}^n=\psi(t_n, x_i, y_j)$ be the fully discrete solutions from \eqref{cq}.
\begin{itemize}
\item[1.]  Both semi-discrete scheme (\ref{tc}) and Euler forward discretization (\ref{cq})  are conservative in the sense that  the total concentration $c_{i,j}$ remains unchanged in time,
\begin{align}
\label{conserve1}
\frac{d}{dt} \sum_{i=1}^{N_x}\sum_{j=1}^{N_y}c^l_{i,j}h^2 & =0, \quad l=1,\cdots, M,\quad \quad t>0\\ \label{conserve2}
\sum_{i=1}^{N_x}\sum_{j=1}^{N_y}c^{l,n+1}_{i,j}h^2 & =\sum_{i=1}^{N_x}\sum_{j=1}^{N_y}c_{i,j}^{l,n} h^2,\quad l=1,\cdots, M .
\end{align}
\item[2.] Assuming $\sigma$ is independent of time and $c^l_{i,j}$ are positive, the semi-discrete {free energy}
\begin{align}\label{DiscretEnergy}
F&=h^2\sum_{l=1}^M\sum_{i=1}^{N_x}\sum_{j=1}^{N_y}\left(c^l_{i,j}\ln c^l_{i,j}+\frac{1}{2}q_lc^l_{i,j}\psi_{i,j}+\frac{1}{2}\alpha_l \left(c^l_{i,j}\right)^2\right) \notag\\
  &+  \frac{h}{2}\sum_{j=1}^{N_y}(\sigma_{1/2,j}\psi_{1,j}+\sigma_{N+1/2,j}\psi_{N,j}) +\frac{h}{2} \sum_{i=1}^{N_x}(\sigma_{i,1/2}\psi_{i,1}+\sigma_{i,N+1/2}\psi_{i,N})
\end{align}

 satisfies
\begin{align}\label{fd}
\frac{d}{dt} F=  &-\frac{\gamma}{h^2}\sum_{l=1}^M\sum_{i=1}^{N_x}\sum_{j=1}^{N_y-1}e^{-q_l(\psi_{i,j+1}+\psi_{i,j})/2-\alpha_l (c^l_{i,j+1}+c^l_{i,j})/2}\left(\ln g^l_{i,j+1}-\ln g^l_{i,j}\right)(g^l_{i,j+1}-g^l_{i,j})\notag\\
               &-\frac{\gamma}{h^2}\sum_{l=1}^M\sum_{i=1}^{N_x-1}\sum_{j=1}^{N_y}e^{-q_l(\psi_{i+1,j}+\psi_{i,j})/2-\alpha_l (c^l_{i+1,j}+c^l_{i,j})/2}\left(\ln g^l_{i+1,j}-\ln g^l_{i,j}\right)(g^l_{i+1,j}-g^l_{i,j})
\notag\\
&\leq 0,
\end{align}
therefore the semi-discrete free energy is non-increasing.
\end{itemize}
\end{theorem}

\begin{proof}
\begin{itemize}
\item[1.]
With the the zero flux boundary conditions $\hat{g}_{x,\frac{1}{2},j}=0$, $\hat{g}_{x,N+\frac{1}{2},j}=0$, $\hat{g}_{y,i,\frac{1}{2}}=0$ and $\hat{g}_{j,i,N+\frac{1}{2}}=0$, summing \eqref{tc} leads to \eqref{conserve1}. Similarly, summing \eqref{pc} leads to \eqref{conserve2}.

\item[2.]  A direct calculation using $\displaystyle \sum_{i=1}^{N_x}\sum_{j=1}^{N_y} \dot{c^l}_{i,j} := \sum_{i=1}^{N_x}\sum_{j=1}^{N_y} \frac{d}{dt} {c^l}_{i,j} =0$ gives
\begin{align}
\frac{d}{dt}{F} &=h^2 \sum_{l=1}^M\sum_{i=1}^{N_x}\sum_{j=1}^{N_y}\left[\left(\ln c^l_{i,j}+q_l\psi_{i,j}+\alpha_l c^l_{i,j}\right)\dot{c^l}_{i,j}+\frac{1}{2}q_l\left( {c^l}_{i,j}\dot{\psi}_{i,j}-\dot{c^l}_{i,j}\psi_{i,j}\right)\right] \notag\\
  &+  \frac{h}{2}\sum_{l=1}^M\sum_{j=1}^{N_y}\sigma(\dot{\psi}_{1,j}+\dot{\psi}_{N,j}) +\frac{h}{2} \sum_{l=1}^M\sum_{i=1}^{N_x}\sigma(\dot{\psi}_{i,1}+\dot{\psi}_{i,N}).
\end{align}
By \eqref{tc}, we further have
\begin{align}
&\left(\ln c^l_{i,j}+q_l\psi_{i,j}+\alpha_l c^l_{i,j}\right)\dot{c^l}_{i,j}\notag\\
=& \frac{\gamma}{h} \ln g^l_{i,j} \left[e^{-(q_l\psi_{i+{\frac{1}{2}},j}+\alpha_l c^l_{i+{\frac{1}{2}},j})}\widehat{g^l}_{x,i+{\frac{1}{2}},j} -e^{-(q_l\psi_{i-{\frac{1}{2}},j}+\alpha_l  c^l_{i-{\frac{1}{2}},j})}\widehat{g^l}_{x,i-{\frac{1}{2}},j}\right] \nonumber\\
&+\frac{\gamma}{h} \ln g^l_{i,j}\left[
e^{-(q_l\psi_{i,j+{\frac{1}{2}}}+\alpha_l  c^l_{i,j+{\frac{1}{2}}})}\widehat{g^l}_{y,i,j+{\frac{1}{2}}} -e^{-(q_l\psi_{i,j-{\frac{1}{2}}}+\alpha_l  c^l_{i,j-{\frac{1}{2}}})}\widehat{g^l}_{y,i,j-{\frac{1}{2}}}\right].\label{F1}
\end{align}
Summing \eqref{F1} over all $l, i,j$ leads to
\begin{align}
&h^2\sum_{l=1}^M\sum_{i=1}^{N_x}\sum_{j=1}^{N_y}\left[\left(\ln c_{i,j}+q\psi_{i,j}+\alpha_l c_{i,j}\right)\dot{c}_{i,j}\right]\notag\\
=&-\frac{\gamma}{h^2}\sum_{l=1}^M\sum_{i=1}^{N_x}\sum_{j=1}^{N_y-1}e^{-q(\psi_{i,j+1}+\psi_{i,j})/2-\alpha_l (c_{i,j+1}+c_{i,j})/2}\left(\ln g_{i,j+1}-\ln g_{i,j}\right)(g_{i,j+1}-g_{i,j})\notag\\
 &-\frac{\gamma}{h^2}\sum_{l=1}^M\sum_{i=1}^{N_x-1}\sum_{j=1}^{N_y}e^{-q(\psi_{i+1,j}+\psi_{i,j})/2-\alpha_l (c_{i+1,j}+c_{i,j})/2}\left(\ln g_{i+1,j}-\ln g_{i,j}\right)(g_{i+1,j}-g_{i,j}),\label{F1all}
\end{align}
where the zero flux boundary conditions are used again.

Using the discrete Poisson equation \eqref{pc}, we have remaining non-boundary terms in $\frac{d}{dt}F$ as
\begin{align}
&\frac{h^2}{2}\sum_{l=1}^Mq_l\left( {c^l}_{i,j}\dot{\psi}_{i,j}-\dot{c^l}_{i,j}\psi_{i,j}\right)\notag\\
=& -\frac{\epsilon}{2}\left[ ({\psi_{i+1,j}-2\psi_{i,j}+\psi_{i-1,j}}) + ({\psi_{i,j+1}-2\psi_{i,j}+\psi_{i,j-1}}) \right]\dot{\psi}_{i,j}\notag\\
 & +\frac{\epsilon}{2}\left[ ({\dot{\psi}_{i+1,j}-2\dot{\psi}_{i,j}+\dot{\psi}_{i-1,j}}) + ({\dot{\psi}_{i,j+1}-2\dot{\psi}_{i,j}+\dot{\psi}_{i,j-1}}) \right]\dot{\psi}_{i,j}.\label{F2}
\end{align}
Summing \eqref{F2} over all $i,j$ leads to
\begin{align}
&\frac{h^2}{2}\sum_{l=1}^M\sum_{i=1}^{N_x}\sum_{j=1}^{N_y}q_l\left( {c^l}_{i,j}\dot{\psi}_{i,j}-\dot{c^l}_{i,j}\psi_{i,j}\right)\notag\\
= & -\frac{h}{2}\sum_{j=1}^{N_y}(\sigma_{1/2,j}\psi_{1,j}+\sigma_{N+1/2,j}\psi_{N,j}) -\frac{h}{2} \sum_{i=1}^{N_x}(\sigma_{i,1/2}\psi_{i,1}+\sigma_{i,N+1/2}\psi_{i,N})\label{F2all}.
\end{align}

Finally the desired \eqref{fd} follows by combining \eqref{F1all} and \eqref{F2all}, and using the fact that $(\ln\alpha-\ln\beta)(\alpha-\beta)\geq 0$ for any $\alpha>0$ and $\beta>0$.

 \end{itemize}
 \end{proof}

{\bf Remark 2.1.} In Theorem \ref{thm:properties}, we proved the conservation and free energy dissipation in 2D, with the assumption of $c_{i,j}>0$. The proof is readily extensible to 3D systems. For the positivity of $c_{i,j}$, we can only theoretically prove it in 1D for a system of single species in Appendix A. Our numerical simulations, however, indicate that the discrete concentrations $c_{i,j}^n$ remain positive in long time for the MPNP system in high dimension with multiple species. To theoretically prove the positivity in high dimensions ($\R^2$ and $\R^3$), it is critical to establish $L^\infty$ bounds for the numerical solutions of concentrations and electrostatic potential. We are currently working on the matter and will report the findings in our future work.

\section{Numerical Simulations}\label{s:NumSim}
\subsection{Numerical Test}
We first consider a closed system with one species of counterion.  Such a system, for instance, describes a membrane with ionizable groups that release one species of ions into aqueous solutions, giving rise to an oppositely charged membrane with the same amount of charges carried by counterions.  We numerically solve the equations \reff{SMPNPi}  on $\Omega = [0, 1] \times [0, 1]$. We take $q=1, \epsilon=1, \gamma=1$, and $\alpha=4$.  The initial and boundary conditions are given respectively by
\begin{align*}
                         c(0, x,y)=2, \quad \psi(0, x,y)=0,
\end{align*}
and
\begin{align*}
\left (\nabla  c+\alpha c \nabla c+  q  c \nabla \psi \right)\cdot \textbf{n}=0, \quad \epsilon \frac{\partial \psi}{\partial {\bf n}}= \begin{cases}
-1 & \text{}  {\{(x,y) |x=1\text{ or }  y=0\}} \\
0  & \text{else,}
\end{cases}
\qquad  \mbox{on} \quad \partial \Omega.
\end{align*}
\begin{figure}[hpbt]
\centering
\subfigure[$c$ for PNP] {\includegraphics[width=0.4\textwidth]{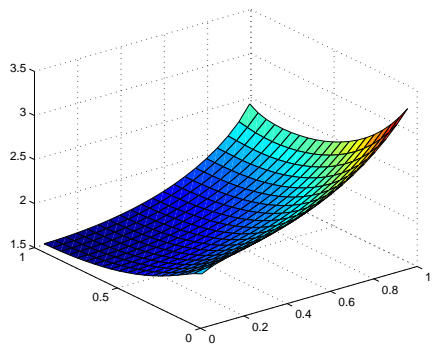}}
	\subfigure[$\psi$ for PNP]{\includegraphics[width=0.4\textwidth]{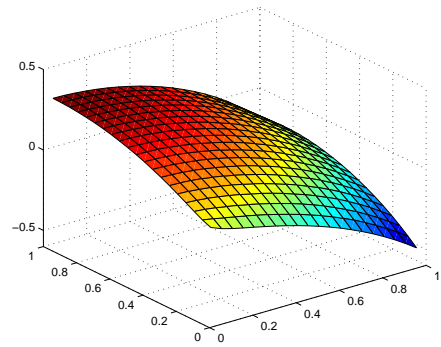}}\\
\subfigure[$c$ for MPNP]{\includegraphics[width=0.4\textwidth]{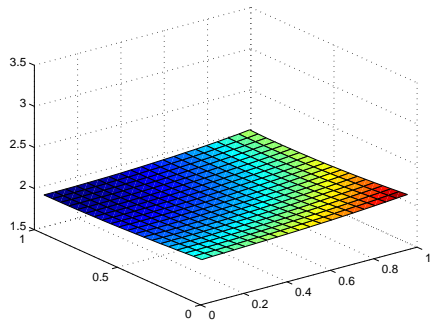}} 
\subfigure[$\psi$ for MPNP]{\includegraphics[width=0.4\textwidth]{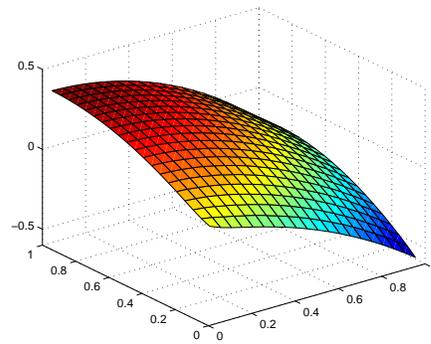}}
\caption{Steady-state solutions of concentration, $c$, and electrostatic potential, $\psi$, for the classical PNP and modified PNP. }
\end{figure}\label{f:SteadyStates2D}
\noindent Note that the results calculated with nonzero $\alpha$ are denoted by the MPNP, and the results obtained with $\alpha=0$ are denoted by the PNP.

Fig.~\ref{f:SteadyStates2D} depicts the steady-state solutions for the classical PNP and MPNP equations. Clearly, we can see that the concentration close to the charge surface for the MPNP is much lower due to the effect of excluded volume of ions. With less ions adsorbed to the charged surface, screening effect stemming from the ions is therefore much weaker, leading to higher electrostatic potential at the charged surface.  This phenomenon indicates that our modified PNP model is able to capture the excluded volume effect of counterions. The result agrees well with other models having the excluded volume effect\cite{BAO_PRL97,Bazant_PRE07I,Bazant_PRE07II, ZhouWangLi_PRE11, LiLiuXuZhou_Nonliearity13, BZLu_JSP16}.  Also, the numerical result agrees with the analysis presented in section \ref{MPB} that the modified Debye screening length $\hat{\lambda}_D^W$ becomes longer due to volume exclusion. We also want to point out that the numerical solutions $c^n_{i,j}$ remain positive in all our simulations for large time, such as $T=5$, which is long after the system becomes steady.

\begin{figure}[htbp]
\centering
\includegraphics[scale=0.6]{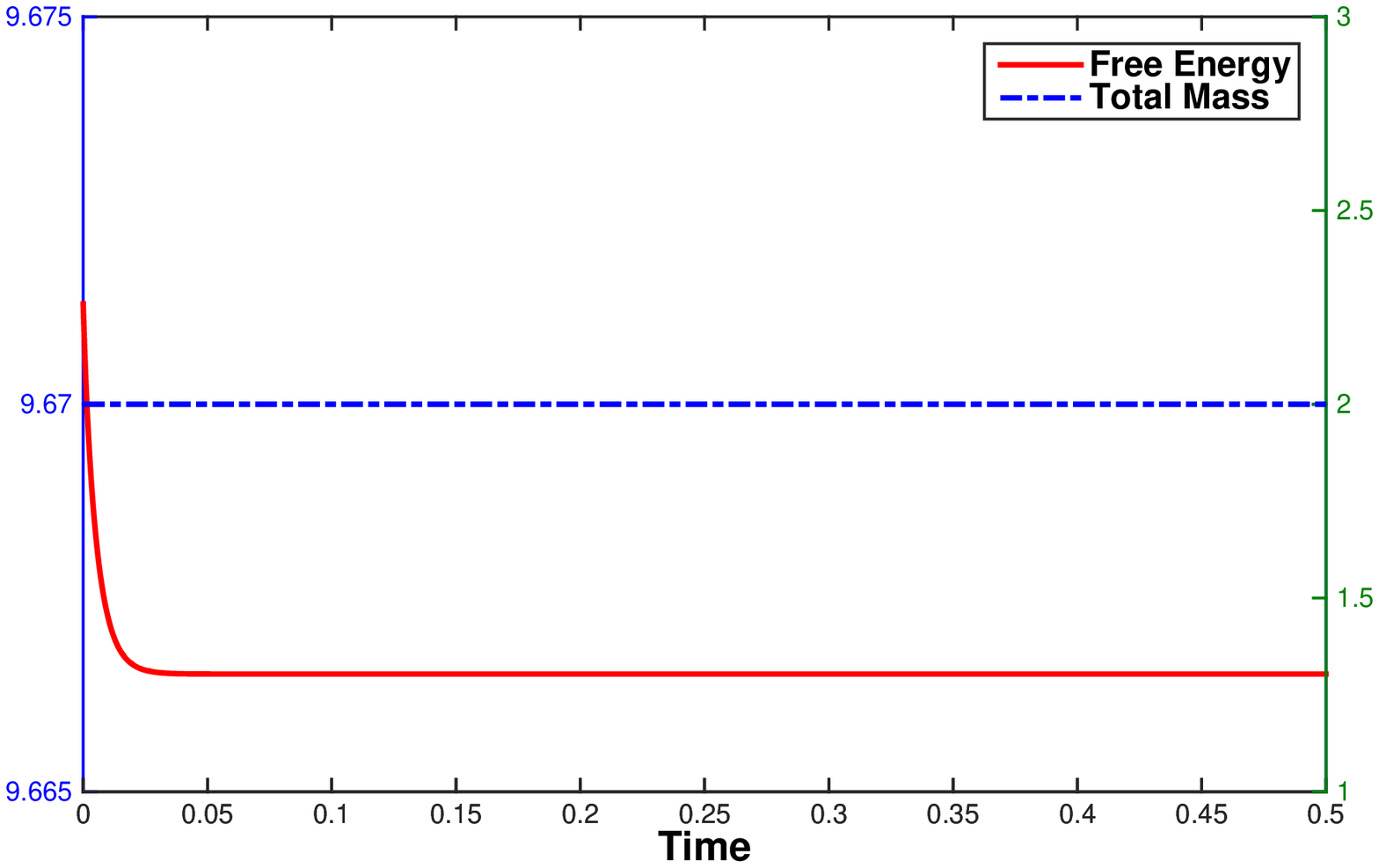}
\caption{Profiles of  the free energy (solid line)  and total concentration (dashed line) for the MPNP equations against time evolution.}  \label{f:EngMass}
\end{figure}

From Fig.~\ref{f:SteadyStates2D}, we have seen that the solutions of concentration are positive on $\Omega$. To test the property of mass conservation, we study the total concentration of the ions with respect to the time evolution. Fig.~\ref{f:EngMass} clearly shows that our numerical scheme perfectly conserves the total concentration. Also, we can see from Fig.~\ref{f:EngMass} that, as time evolves, the energy \reff{DiscretEnergy} decays monotonically and robustly.  Overall, such results confirm our numerical analysis presented in section \ref{s:Properties} on the properties of mass conservation and energy dissipation.

\begin{table}[htbp]
\centering
\begin{tabular}{|c|c|c|c|c|} \hline
h &      $l^{\infty}$ error in $c$ & Order & $l^{\infty}$ error in $\psi$    &Order      \\   \hline
 0.25&   0.0013665  &         --&   0.00017634   &     --    \\ \hline
          0.2&  0.00087098  &     2.0182&   0.00012124   &    1.6788    \\   \hline
          0.1 & 0.00020778   &    2.0676 & 3.3199e-005    &   1.8687    \\   \hline
         0.05 &4.1572e-005   &    2.3214 & 7.0764e-006    &   2.2300    \\   \hline
                          \end{tabular}
                            \caption{The $l^\infty$ error and convergence order for $c$ and $\psi$.}  \label{t:2d1errortable}
\end{table}

To test the accuracy of our  numerical  scheme, we solve the problem with various spatial step size $h$ and temporal step sizes $k$, with $k= \mathcal{O}(h^2)$. Table \ref{t:2d1errortable} lists the $l^{\infty}$ errors and their convergence orders.  A reference solution with a highly refined mesh is used to calculate the errors, since the exact solution is not available in this case.  In Table \ref{t:2d1errortable}, we observe that the $l^{\infty}$ error decreases as the mesh is refined. The convergence order is around $2$ for both the concentration and electrostatic potential, which implies that our numerical scheme has expected accuracy, i.e., second-order accurate in spatial discretization and first-order accurate in temporal discretization.

\subsection{Charge Dynamics}\label{ss:Dyn}
\begin{figure}[htbp]
\centering
\includegraphics[width=0.7\textwidth,height=0.28\textheight]{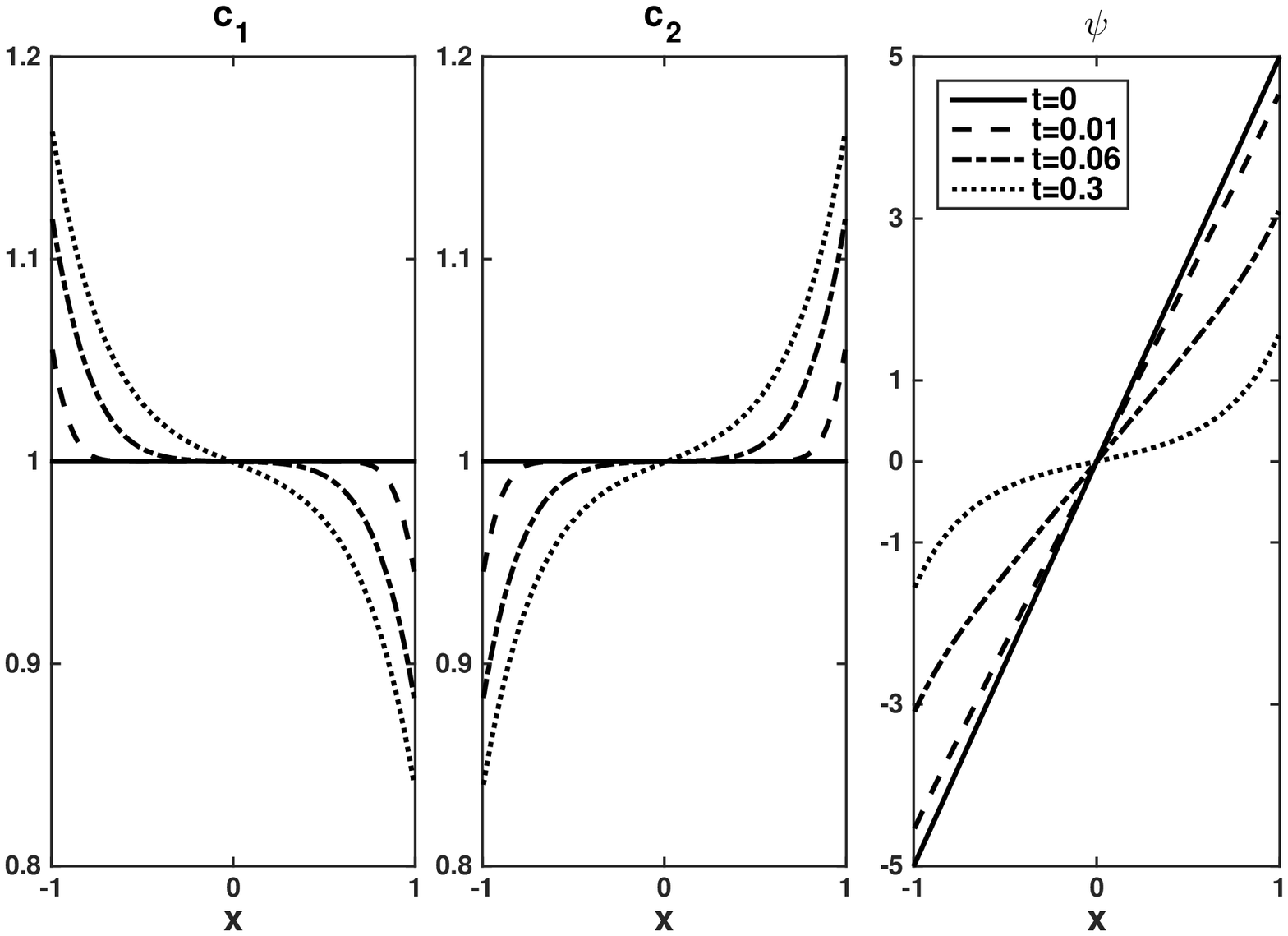}\\
\vspace{-3mm}\includegraphics[width=0.7\textwidth,height=0.28\textheight]{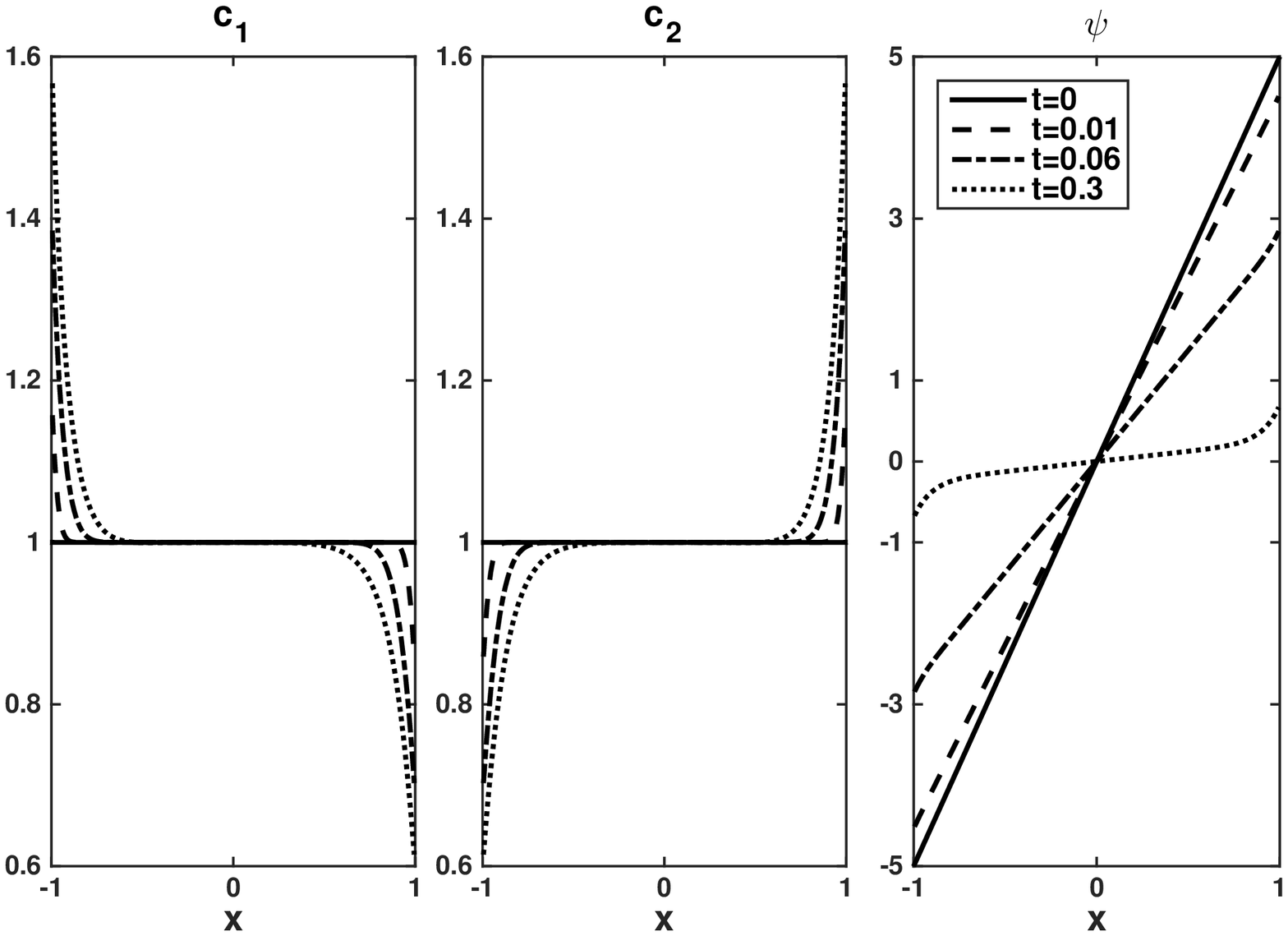}
\caption{Dynamics of concentrations and electrostatic potential for the MPNP (upper panel) and PNP (lower panel) equations.} \label{CPsiEvo}
\end{figure}

\begin{figure}[h]
 \centering
\includegraphics[width=0.7\textwidth,height=0.28\textheight]{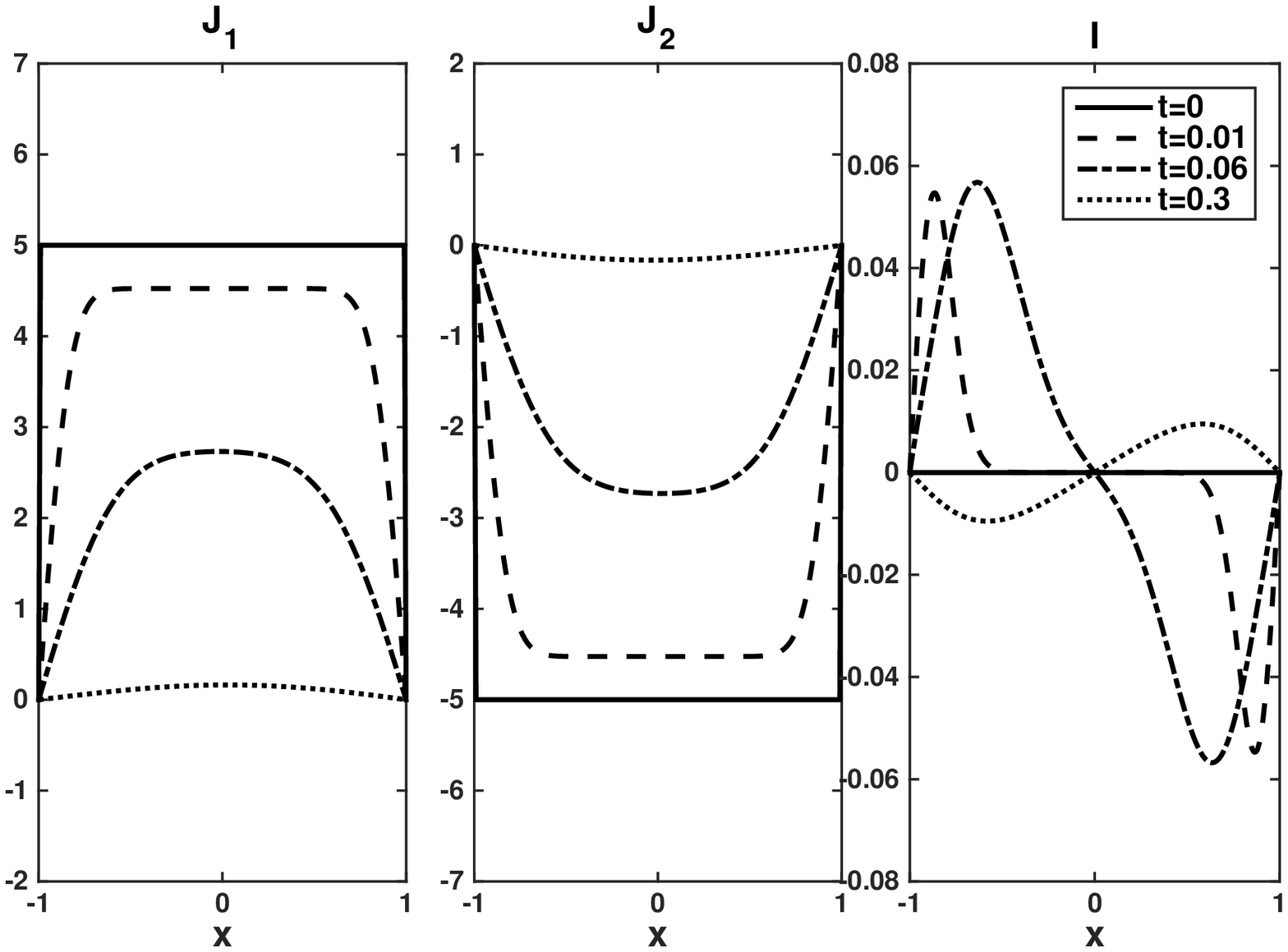}\\
\vspace{-3mm}\includegraphics[width=0.7\textwidth,height=0.28\textheight]{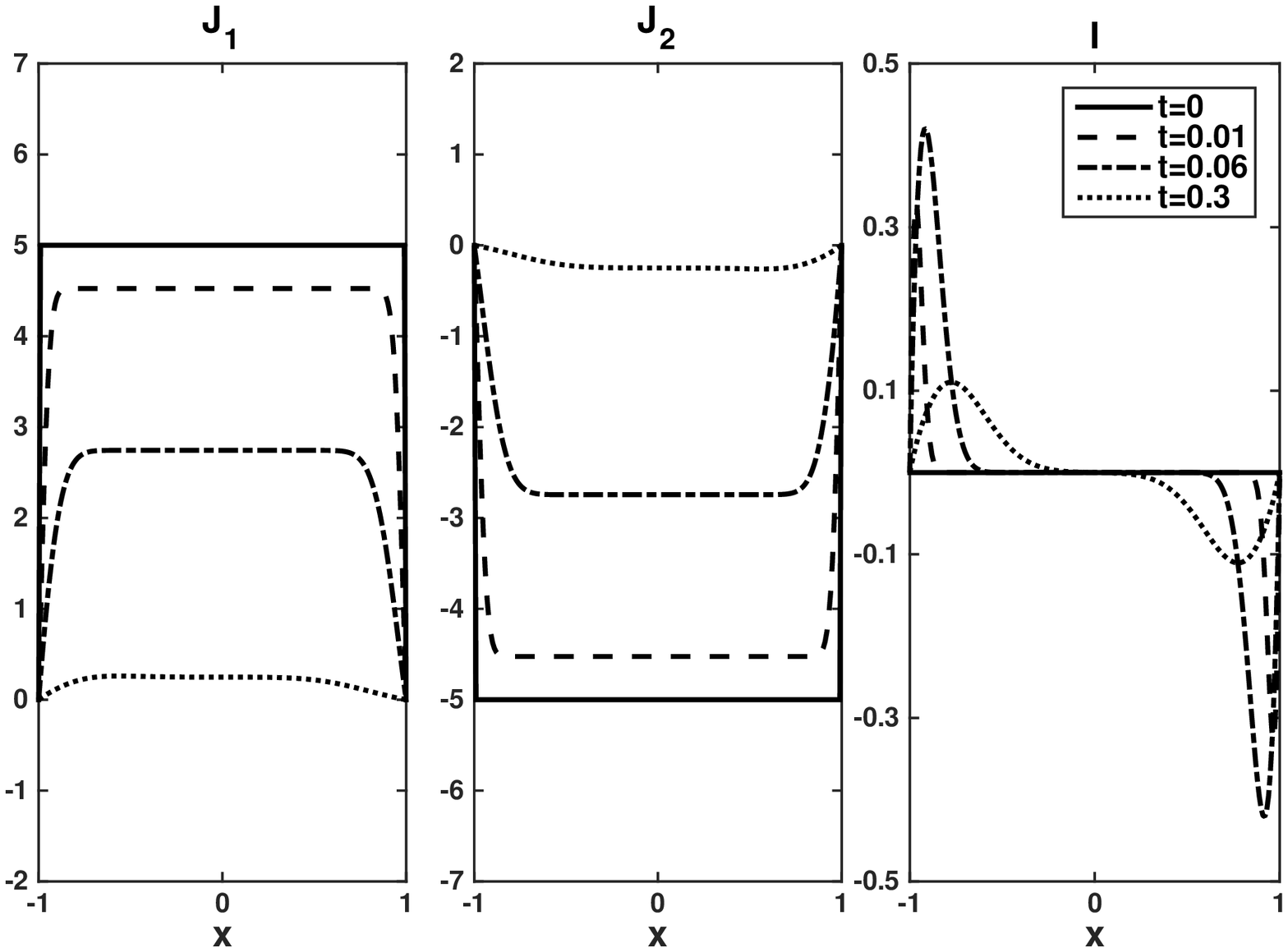}
\caption{Dynamics of flux of the MPNP (upper panel) and PNP (lower panel) equations. The flux for each species is defined by $J_l=-\left(\partial_x c^l+\alpha_l c^l \partial_x c^l+  q_l  c^l \partial_x \psi\right) $, and the sum of flux $I$ is given by $I=J_1+J_2$. } \label{FluxEvo}
\end{figure}
To study the charge dynamics of the MPNP equations, we consider a closed, neutral system that consists of two large parallel blocking surfaces with surface charges and two species of ions. We assume that the system is homogeneous in $y$ and $z$ directions. The equations \reff{SMPNPi} are reduced to one dimension. We set $M=2, q_1=-q_2=1, \alpha_1=\alpha_2=8, \epsilon=0.02, \gamma=0.1$, and initial and boundary conditions
\begin{align*}
&c_1(x, 0)=1, \quad c_2(x, 0)=1,  \quad - \epsilon \partial_x\psi(t,-1)=\sigma_a= - 0.1, \quad  \epsilon \partial_x \psi(t,1)=\sigma_b= 0.1, \\
&\left.\left( \partial_x c^l+\alpha_l c^l \partial_x c^l+  q_l  c^l \partial_x \psi \right)\right|_{x=-1,1}=0 \quad \mbox{for} \quad  l=1,2.
\end{align*}
We study the dynamics of concentrations and potential in an applied electric field induced by two charged surfaces. From Fig.~\ref{CPsiEvo}, we can see that the surface charges attract oppositely charged ions both for the MPNP and PNP equations, and that electrostatic potential at the surfaces decreases due to the screening effect from adsorbed counterions. Comparing with the results of the MPNP equations, the ionic concentrations at the vicinity of surfaces are much higher for the PNP equations, because counterions can accumulate at the charged surfaces without steric hindrance. Therefore, the electrostatic potential at the surfaces for the PNP is lower due to stronger screening effect.

It is of interest to study the excluded volume effect on the dynamics of flux for each ionic species.  As shown in Fig.~\ref{FluxEvo}, each species has large flux between charged surfaces and gradually relaxes to zero, reaching an equilibrium. In contrast to the results of the MPNP, the sum of flux, $I$, for the PNP has a larger magnitude due to its ignorance of excluded volume effect of ions. During the charge diffusion, the sum of flux for the MPNP in the middle region grows much faster than that of the PNP, indicating that the excluded volume effect speeds up the transport of ions through collisions between ions. Therefore, the system reaches an equilibrium in a smaller timescale if the steric effect is taken into account.

\subsection{Effect of $\alpha_l$}\label{ss:alpha}
\begin{figure}[hbt]
\centering
\includegraphics[width=0.8\textwidth]{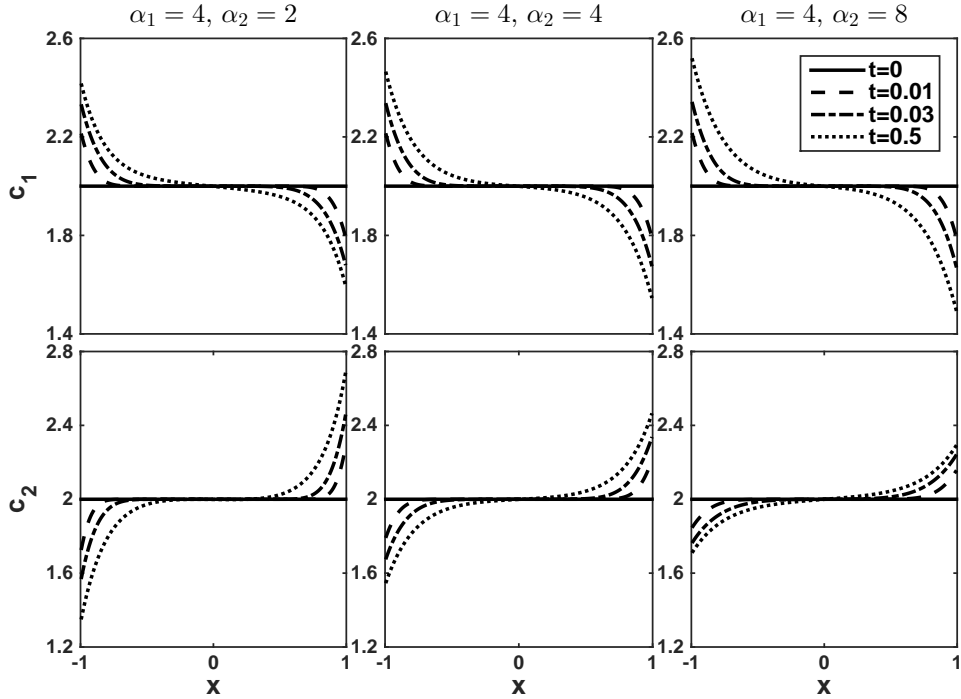}
\caption{Effect of $\alpha_1$ and $\alpha_2$ on the dynamics of concentrations $c_1$ (upper panel) and $c_2$ (lower panel).} \label{EffAlpha}
\end{figure}
As discussed above, the excluded volume effect that is reflected by $\alpha_i$ has a pronounced impact on the dynamics of charge diffusion. The value of $\alpha_i$ is related to the size of each species of ions. It can be understood as a fitting parameter, either in the modification of diffusion constants of ions\cite{Bruna_PRE12} or the sizes of ions in the Lennard-Jones potential\cite{LinBob_CMS14}.  We vary the value of $\alpha_i$ and investigate its effect on the dynamics and equilibrium concentration distributions.

In our simulations, we use the same setting as the previous section, except that $c_1(x, 0)= c_2(x, 0)=2$ and $\sigma_b=- \sigma_a= 0.2$. First we study the effect of $\alpha_2$ by testing different values of $\alpha_2$ ($2$, $4$, and $8$) with fixed $\alpha_1$. From Fig.~\ref{EffAlpha}, we observe that the dynamics of the concentration $c_2$ change significantly as $\alpha_2$ grows. Because of the steric hindrance, larger ionic sizes result in lower ionic concentration adsorbed to charge surfaces. In addition, ions with larger sizes reach an equilibrium much faster due to more frequent collisions between particles. With less accumulated counterions at surfaces, for instance $c_2$ at the right charged surface, electrostatic potential is less screened and therefore has stronger repulsion against coions (i.e., $c_1$). As such, we can see that the minor effect of $\alpha_2$ on $c_1$ is mainly through the variation of electrostatic potential. Direct interactions between $c_1$ and $c_2$ can be taken into account by including cross diffusion between different ionic species. We defer this investigation to our future work.

\subsection{Timescales in Charge Diffusion}\label{ss:ChgDyn}
\begin{figure}[hbt]
\centering
\includegraphics[width=0.8\textwidth]{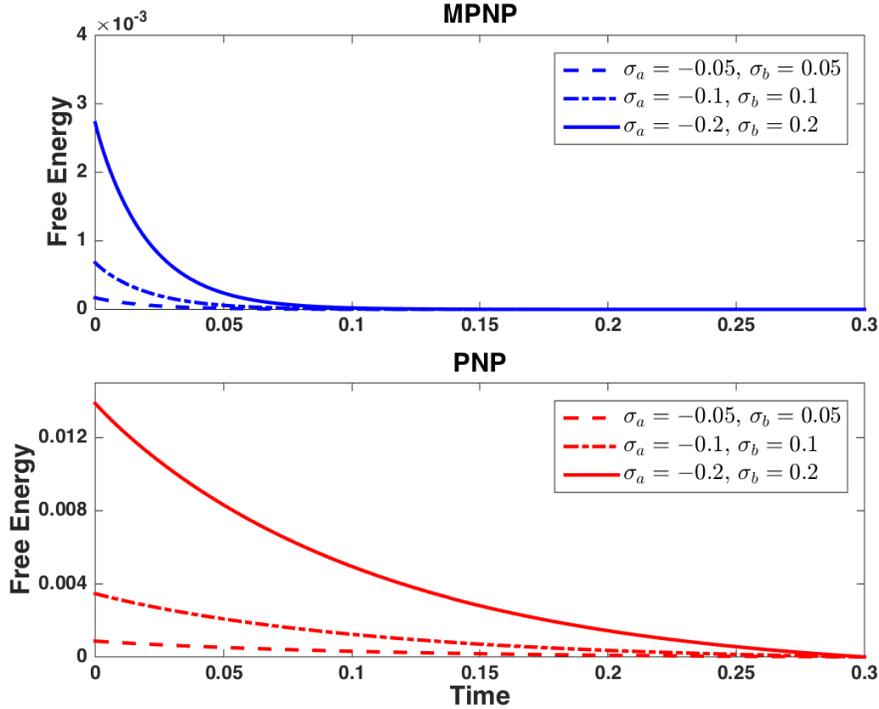}
\caption{Energy decay for the MPNP and PNP.} \label{energydecay}
\end{figure}
As revealed in previous two examples, there is a significant difference in timescales of relaxation dynamics whether the excluded volume effect is included or not. In this case, we probe the relaxation timescales in the charge diffusion through analyzing the free energy and total diffuse charges. We consider a system with the same setting as in section \ref{ss:Dyn} except $\epsilon=1, \gamma=1, c_1(x, 0)= c_2(x, 0)=1$, and $\alpha_1=\alpha_2=8$.  As expected, Fig.~\ref{energydecay} displays monotone energy profiles against time steps. For ease of reading, we shift each energy profile by the free energy of its final equilibrium state.  When larger surface charge is applied, the energy difference between the initial state and the equilibrium state is much higher, implying that more energy is stored in adsorbed counterions.  It is easy to notice that the energy for the MPNP relaxes quickly to a constant value for $T>0.1$; whereas, the energy for the PNP decreases gradually with a long tail. Such a discrepancy clearly demonstrates that the relaxation process for the MPNP is much faster than that of the PNP.  This can be explained by the fact that the excluded volume effect contributes to the diffusion of the ionic concentration through particle collisions and therefore promotes the energy relaxation of the whole system.

\begin{figure}[hbt]
\centering
\includegraphics[width=0.8\textwidth]{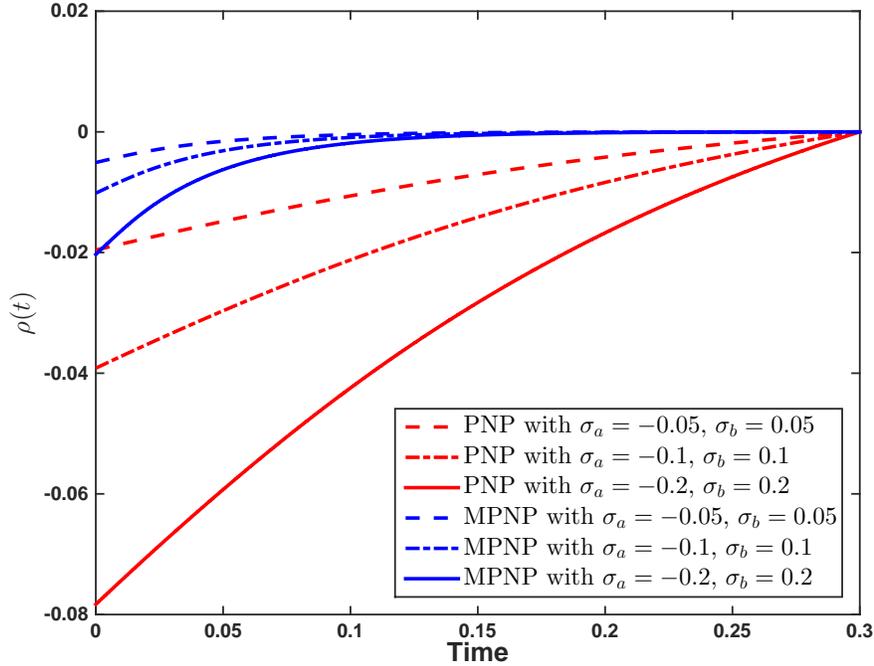}
\caption{ Total diffuse charges $\rho(t)$ for the modified PNP and  classical PNP.} \label{ex3rhot}
\end{figure}
 To further understand the timescales of charge diffusion, we also study the evolution of total diffuse charges in left half of the electrolytes \cite{BazantDiffChg_PRE04}:
\[
\rho(t)=\int_{-1}^0\sum_{l=1}^mq_lc^l (x) dx.
\]
From Fig.~\ref{ex3rhot}, we find that the total diffuse charges for the MPNP increases quickly and reaches a plateau. In contrast, the total diffuse charges in the PNP keeps growing over a relatively long period. As such, the timescale of the charge diffuse for the MPNP is much smaller.  Again, collisions between ions with excluded volume effect accounts for the smaller timescale exhibited in charge diffusion modeled by the MPNP.  Also, the MPNP successfully predicts many less diffuse charges, since the charges carried by ions are sterically hindered from adsorbing to the surface. All the results demonstrate that the MPNP theory has effectively captured the excluded volume effect of ions.

\section{Conclusions and Discussions}\label{s:Con}
In this work, we proposed a simple yet effective modified PNP (MPNP) model with the excluded volume effect. Our model used a linear concentration-dependent diffusion coefficient to incorporate the excluded volume effect of ions. With the help of the Lambert-W special function, we obtained the corresponding modified Poisson-Boltzmann (MPB)  equation for the steady state. A further analysis showed that there exists a unique weak solution to the MPB equation. In the limit of weak electrostatic potential, the MPB is approximated by a Debye-H\"uckel equation with a modified Debye screening length, which is longer than the classical one.  This prediction agrees with other modified PB models in the literature and is later confirmed by our numerical simulations. In the limit of strong electrostatic potential, the MPB is approximated by a linearized equation with a different modified Debye screening length that depends on parameters arising from volume exclusion.

In addition, we developed an accompanying conservative and energy dissipative finite difference method for the proposed MPNP model. Our analysis confirmed that the numerical scheme conserves total concentration and satisfies a corresponding discrete energy dissipation law. Positivity of numerical solutions was  proved for a system with single species in 1D. Numerical experiments were conducted to demonstrate that the scheme is of second-order accurate in spatial discretization and has expected properties. Extensive numerical simulations revealed that the excluded volume effect of ions has significant impacts on the dynamics of ionic concentration and flux. From the evolution of free energies and diffuse charges, we found that the excluded volume effect leads to a decrease of the timescales of charge diffusion through ionic collisions.

We now discuss several issues and possible further refinements of our work. In our current model, cross interactions between different species that arise from the excluded volume effect have not been taken into account. The model is effective when there is only one species in the environment, such as counterions adsorbing to charged surfaces. When multiple species of ions present, the cross interactions can be considered by including nonlinear cross-diffusion terms in the Nernst-Planck equations\cite{Bruna_JCP12, HyonLiuBob_CMS10, LinBob_CMS14}. It is interesting to explore the impact of cross interactions on the dynamics of ions. The corresponding numerical schemes that have properties of mass conservation, solution positivity, and energy dissipation will be one of our future studies as well.

As proved in Theorem \ref{thm:properties}, we can show that our numerical scheme respects ionic mass conservation and energy dissipation. In our numerical examples, we have numerically verified that the numerical solutions of concentrations keeps being positive in long time simulations. Unfortunately, we are not able to rigorously prove the positivity of numerical solutions of concentrations except for the 1D case, see Appendix \ref{s:App}. The main difficulty lies in the establishment of $L^\infty$ bounds for the numerical solutions of electrostatic potential and concentrations.  One possible improvement is to design a novel discretization scheme for the Nernst-Planck equations, so that the positivity of numerical solutions of concentrations can be guaranteed.

Finally, it is of great interest to develop implicit schemes for the MPNP model. In our current implementation, the Nernst-Planck equations are discretized  explicitly and the Poisson's equation is solved in each time step. The discretization time step has to be small for stability reasons. This treatment becomes computationally inefficient for 3D cases. In future, we will focus on the development of implicit schemes that have mass conservation, positivity of numerical solutions, and energy dissipation.

}
\vspace{0.5cm}
\noindent{\bf Acknowledgments.}
S. Zhou acknowledges the supports from Soochow University through a start-up Grant (Q410700415), National Natural Science Foundation of China (NSFC 11601361), and Natural Science Foundation of Jiangsu Province (BK20160302).

\appendix
\numberwithin{equation}{section}
\section{Appendix: Positivity in 1D}\label{s:App}
We investigate the positivity of concentration $c$ in single species system varying only in one direction. This kind of system applies to many situations, e.g., two parallel blocking plates with charged surfaces shown in Example 2.
\begin{theorem}\label{thm:positivity} Assume the system \eqref{SMPNPi} with single species, $M=1$, is varying only in $y-$direction, i.e., $c(i,j)=c_j$ and $\psi(i,j)=\psi_j$.
 The discrete concentration $c_j^n$ remains  positive in time:  if  $c_j^n > 0$,  then
$$
c_j^{n+1} > 0$$
provided the condition $k < h^2 \lambda_0/\gamma $ where
\begin{equation}\label{la}
\lambda_0= \frac{e^{\frac {\alpha (\sigma_a+\sigma_b)}{hq^2}}}{ e^{\frac{-hq^2\sigma_b}{2}} + e^{\frac{-hq^2\sigma_a}{2}}}.
\end{equation}

\end{theorem}
\begin{proof}
 Define $A^n_j=q\psi_{j+1}^n-q\psi_j^n$,  the boundary condition gives $A_0^n=-qh\sigma_a/\epsilon$ and $A_N^n=qh\sigma_b/\epsilon$. Let mesh ratio be denoted by $\lambda=k\gamma/h^2$,  we can rewrite (\ref{cq}) as

\begin{align}
c_j^{n+1}
=&c^n_j\left( 1- \lambda \left( e^{\frac{-q(\psi_{j+1}^{n}-\psi_{j}^{n})}{2}} e^{ \frac{-\alpha (c_{j+1}^{n}-c_{j}^{n})}{2}} + e^{\frac{q(\psi_{j}^{n}-\psi_{j-1}^{n})}{2}} e^{ \frac{\alpha (c_{j}^{n}-c_{j-1}^{n})}{2}} \right) \right)\notag\\
&+\lambda c_{j+1}^n e^{\frac{q(\psi_{j+1}^{n}-\psi_{j}^{n})}{2}} e^{ \frac{\alpha (c_{j+1}^{n}-c_{j}^{n})}{2}} +
\lambda  c_{j-1}^n e^{\frac{q(-\psi_{j}^{n}-\psi_{j-1}^{n})}{2}} e^{ \frac{-\alpha (c_{j}^{n}-c_{j-1}^{n})}{2}}  .
\label{cnplus1}\end{align}
      As in \cite{LW14}, the discrete Poisson equation implies
\begin{align}
          A_j^n-A^n_{j-1}&=-q^2c_j^nh^2/\epsilon. \label{Abound}
\end{align}
This indicates $A_j^n$ is monotonic. Along with boundary conditions we have
 $A_j^n$ is bounded, i.e.,
$$
h\sigma_b \leq A_j^n \leq -h\sigma_a.
$$
Furthermore,  the discrete Poisson equation implies
\begin{align}
 h^2 q^2(c_{j}^{n}-c_{j-1}^{n})&=-(A_{j}^{n}-2A_{j-1}^{n}+A_{j-2}^{n}).\label{cinAs}
\end{align}
Combining \eqref{Abound} and \eqref{cinAs}, we have the following for \eqref{cnplus1}



\begin{align}
e^{ \frac{ q(\psi_{j}^{n}-\psi_{j-1}^{n})}{2}}e^{ \frac{\alpha (c_{j}^{n}-c_{j-1}^{n})}{2}}+e^{ -\frac{ q(\psi_{j+1}^{n}-\psi_{j}^{n})}{2}}e^{ \frac{-\alpha (c_{j+1}^{n}-c_{j}^{n})}{2}}&\leq e^{ \frac{hq^2\sigma_a}{2}-\frac{\alpha (\sigma_a+\sigma_b)}{hq^2}}+e^{ -\frac{hq^2\sigma_b}{2}-\frac{\alpha (\sigma_a+\sigma_b)}{hq^2}}\notag\\&=e^{-\frac{\alpha (\sigma_a+\sigma_b)}{hq^2}}(e^{ -\frac{hq^2\sigma_a}{2}}+e^{ -\frac{hq^2\sigma_b}{2}})
\end{align}
Thus,  we have $ c_j^{n+1}> 0$ if  $\lambda \leq \lambda_0$ as  defined in (\ref{la}).
 \end{proof}

\bibliographystyle{plain}
\bibliography{PNP}
\end{document}